\documentclass[12pt,a4paper]{article}
\usepackage{amsmath,amssymb,amsthm} 
\newtheorem{lemma}{Lemma}
\newtheorem{thm}{Theorem}

\newtheorem{cor}{Corollary}
\theoremstyle{remark}
\newtheorem{remark}{Remark}
\newcommand{\beq}{\begin{equation}}
\newcommand{\eeq}{\end{equation}}
\newcommand{\rd}{\partial}

\newcommand{\tp}[1]{\,{\vphantom{#1}}^\mathrm{t}\!\,#1}
\newcommand{\CC}{\mathbb{C}}
\newcommand{\PP}{\mathbb{P}}
\newcommand{\RR}{\mathbb{R}}
\newcommand{\ZZ}{\mathbb{Z}}
\newcommand{\calL}{\mathcal{L}}
\newcommand{\calM}{\mathcal{M}}
\newcommand{\calP}{\mathcal{P}}
\newcommand{\calT}{\mathcal{T}}
\newcommand{\calW}{\mathcal{W}}
\newcommand{\rmS}{\mathrm{S}}
\newcommand{\bsp}{\boldsymbol{p}}
\newcommand{\bst}{\boldsymbol{t}}
\newcommand{\bsx}{\boldsymbol{x}}
\newcommand{\bszero}{\boldsymbol{0}}

\begin{document}

\title{Cubic Hodge integrals and integrable hierarchies of Volterra type}
\author{Kanehisa Takasaki\thanks{E-mail: takasaki@math.kindai.ac.jp}\\
{\normalsize Department of Mathematics, Kindai University}\\ 
{\normalsize 3-4-1 Kowakae, Higashi-Osaka, Osaka 577-8502, Japan}}
\date{}
\maketitle

\begin{abstract}
A tau function of the 2D Toda hierarchy can be obtained from 
a generating function of the two-partition cubic Hodge integrals. 
The associated Lax operators turn out to satisfy an algebraic relation. 
This algebraic relation can be used to identify a reduced system 
of the 2D Toda hierarchy that emerges when the parameter $\tau$ 
of the cubic Hodge integrals takes a special value. 
Integrable hierarchies of the Volterra type are shown to be 
such reduced systems. They can be derived for positive rational 
values of $\tau$. In particular, the discrete series $\tau = 1,2,\ldots$ 
correspond to the Volterra lattice and its hungry generalizations. 
This provides a new explanation to the integrable structures 
of the cubic Hodge integrals observed by Dubrovin et al. 
in the perspectives of tau-symmetric integrable Hamiltonian PDEs. 
\end{abstract}


\section{Introduction}

The Hodge integrals 
\[
  \int_{\overline{\calM}_{g,n}}\lambda_1^{j_1}\cdots\lambda_g^{j_g}
    \psi_1^{k_1}\cdots\psi_n^{k_n},\quad 
  j_1,\ldots,j_g,k_1,\ldots,k_n \geq 0, 
\]
are intersection numbers of two kinds of tautological 
cohomology classes $\psi_1,\ldots,\psi_n$ (the $\psi$-classes) 
and $\lambda_1,\ldots,\lambda_g$ (the Hodge classes) 
on the Deligne-Mumford moduli space $\overline{\calM}_{g,n}$ 
of connected complex stable curves.  
Special linear combinations of these integrals of the form 
\beq
  \int_{\overline{\calM}_{g,n}}
  \frac{\prod_{i=1}^d\Lambda_g^\vee(a_i)}
       {\prod_{i=1}^n(1 - z_i\psi_i)},\quad 
  \Lambda^\vee_g(u) = u^g - u^{g-1}\lambda_1 + \cdots + (-1)^g\lambda_g, 
  \label{specialHI}
\eeq
appear in computation of Gromov-Witten invariants 
of a $d$-dimensional manifold by the method 
of localization \cite{GP99}. 

A combinatorial expression of the cubic (i.e., $d = 3$) 
special Hodge integrals (\ref{specialHI}) was proposed 
by the Gokakumar-Mari\~{n}o-Vafa conjecture \cite{GV99,MV01}. 
The Hodge integrals in this formula are organized 
into an all-genus generating function $G_\mu(a_1,a_2,a_3)$ 
that depends on an integer partition $\mu$, hence called 
the one-partition Hodge integrals.  The parameters $a_i$  
are required to satisfy the so called 
Calabi-Yau condition $a_1 + a_2 + a_3 = 0$. 
Moreover, since the generating function enjoys 
scale invariance, $a_i$'s can be effectively parametrized 
by a single parameter $\tau \not= 0,1$ as 
\[
  a_1 = 1, \quad a_2 = \tau,\quad a_3 = - \tau - 1. 
\]
Let $G_\mu(\tau)$ denote $G_\mu(a_1,a_2,a_3)$ 
in this parametrization.  Liu, Liu and Zhou proved 
the Gopakumar-Mari\~{n}o-Vara conjecture \cite{LLZ0306} 
\footnote{Okounkov and Pandharipande 
presented another proof \cite{OP0307}.}
and extended it to a two-partition version 
$G_{\mu\bar{\mu}}(\tau)$ \cite{LLZ0310272}.  
Moreover, Zhou pointed out that the KP and 2D Toda 
hierarchies underlie these combinatorial expressions 
of the cubic Hodge integrals \cite{Zhou0310408}. 

Some ten years after the work of Liu et al.,  
Dubrovin, Zhang and their collaborators addressed 
the cubic Hodge integrals in the perspectives 
of integrable Hamiltonian PDEs 
and quasi-Miura transformations \cite{DLYZ1409}.  
They observed that integrable hierarchies 
of the Volterra type, typically the Volterra lattice 
\cite{KvM75} (referred to as the discrete KdV hierarchy there), 
emerge in the cubic Hodge integrals for particular values 
of the parameters \cite{DY1606,DLYZ1612,LZZ17}. 

In this paper, we return to Zhou's work \cite{Zhou0310408} 
and elucidate an origin of the integrable structures 
of the Volterra type observed by Dubrovin et al.  
The two-partition Hodge integrals are organized 
to a generating function of two copies of 
the power sum variables $p_i,\bar{p}_i$, $i = 1,2,\ldots$. 
Zhou used a fermionic expression of the generating function 
to show that the generating function 
(slightly modified to depend on the lattice coordinate $s$
\footnote{The meaning of $s$ in the context of Hodge integrals 
is obscure.}) 
becomes a tau function of the 2D Toda hierarchy.  
We consider the associated two Lax operators $L,\bar{L}$ 
of the 2D Toda hierarchy.  These Lax operators turn out 
to satisfy an algebraic relation (Theorem \ref{key-thm}).  
This implies the existence of an auxiliary Lax operator $\calL$ 
(Corollary \ref{key-cor}).  

The operator $\calL$ plays a central role 
to identify an underlying integrable structure. 
When $\tau$ is a positive integer, $\calL$ coincides 
with the Lax operator of a generalized Volterra lattice called 
the Bogoyavlensky-Itoh-Narita (aka the hungry Lotka-Volterra) 
system \cite{Narita82,Bog87,Itoh87}.  The case of $\tau = 1$ 
corresponds to the the Volterra lattice as proved 
by Dubrovin et al. \cite{DLYZ1612} by a different method. 
When $\tau$ is a positive rational number, 
$\calL$ becomes the Lax operator of a generalization 
of the Bogoyavlensky-Itoh-Narita system.  
This amounts to the rational Volterra hierarchy 
introduced by Liu et al. \cite{LZZ17}.  
We can thus explain the origin of integrable hierarchies 
of the Volterra type in a unified way.  
On the other hand, when $\tau$ is a negative rational number, 
we encounter a different kind of integrable structures, 
namely, a lattice version of the Gelfand-Dickey hierarchy 
\cite{Frenkel95} and further reductions thereof. 

We thus find the following integrable structures 
as reductions of the master 2D Toda hierarchy 
(or one of its lattice KP sectors). 
\begin{itemize}
\item 
$\tau = N$, $N = 1,2,\ldots$: 
The relevant integrable structure is the $N+1$-step 
Bogoyavlensky-Itoh-Narita system. The first member 
of this series is the Volterra lattice. 
\item 
$\tau = b/a$, $a$ and $b$ are positive coprime integers: 
The relevant integrable structure is a generalization 
of the Bogoyavlensky-Itoh-Narita system. 
\item 
$\tau = - b/a$, $a$ and $b$ are positive coprime integers: 
A lattice version of the Gelfand-Dickey hierarchy emerges 
in a further reduced form.  A particularly interesting subset 
is the discrete series $\tau = -N/(N+1)$, $N = 1,2,\ldots$.  
\end{itemize}

To prove the algebraic relation of the Lax operators, 
we use a method developed in our previous work 
on the melting crystal model 
and a family of topological string theory 
\cite{Takasaki13a,Takasaki13b,Takasaki14}. 
This method is based on a factorization problem 
that characterizes the dressing operators  
behind the Lax operators. We use it to compute 
{\it the initial values} of the Lax operators.  
As it turns out, these initial values of the Lax operators 
satisfy the algebraic relation in question. 
We can then conclude, by a simple reasoning, 
that the algebraic relation is also satisfied 
throughout the time evolutions.  

This paper is organized as follows.  
Section 2 reviews the combinatorial expression 
of the two-partition Hodge integrals.  
Relevant geometric and combinatorial notions 
are introduced here. 
Section 3 recalls the construction of the tau function 
from the two-partition Hodge integrals.  
A 2D complex free fermion system is a fundamental tool 
of this section.  
Section 4 presents the key theorem and its proof. 
Fermionic expressions of building blocks of the tau function 
are translated to the language of difference operators.  
Those difference operators play a central role 
in this and the next sections.  
Section 5 is devoted to various integrable structures 
that emerge when $\tau$ takes special values.

\section{Two-partition Hodge integrals}

\subsection{Geometric definition} 

Let $G_{g\mu\bar{\mu}}(\tau)$ denote the following 
cubic Hodge integral \cite{LLZ0310272} that depend 
on a parameter $\tau \not= 0$ and two partitions
\footnote{All partitions in this paper are understood 
to have no restriction on the length.  
Such a partition $\mu$ is represented 
by a non-increasing sequence $(\mu_i)_{i=1}^\infty$, 
$\mu_1 \geq \mu_2 \geq \cdots$, of non-negative integers 
such that $\mu_i = 0$ for all $i$'s greater than a bound $n$. 
The minimum of the bound $n$ is the length of $\mu$, 
and denoted by $l(\mu)$.}
$\mu = (\mu)_{i=1}^\infty$, 
$\bar{\mu} = (\bar{\mu}_i)_{i=1}^\infty$, 
$(\mu,\bar{\mu}) \not= (\emptyset,\emptyset)$: 
\beq
  G_{g\mu\bar{\mu}}(\tau) 
  = c_{\mu\bar{\mu}}(\tau) 
   \int_{\overline{\calM}_{g,l(\mu)+l(\bar{\mu})}}
   \frac{\Lambda^\vee_g(1)\Lambda^\vee_g(\tau)\Lambda^\vee_g(-\tau-1)}
    {\prod_{i=1}^{l(\mu)}\frac{1}{\mu_i}(\frac{1}{\mu_i} - \psi_i)
     \prod_{i=1}^{l(\bar{\mu})}\frac{\tau}{\bar{\mu}_i}
     (\frac{\tau}{\bar{\mu}_i} - \psi_{l(\mu)+i}) }. 
  \label{Ggmmbar(tau)}
\eeq
$c_{\mu\bar{\mu}}(\tau)$ is a combinatorial factor of the form 
\[
\begin{aligned}
  c_{\mu\bar{\mu}}(\tau)
  &= - \frac{\sqrt{-1}^{l(\mu)+l(\bar{\mu})}}{z(\mu)z(\bar{\mu})}
       (\tau(\tau + 1))^{l(\mu)+l(\bar{\mu})-1}\\
  &\quad\mbox{}\times 
  \prod_{i=1}^{l(\mu)}\frac{(\mu_i\tau)_{\mu_i-1}}{\mu_i!}
  \prod_{i=1}^{l(\bar{\mu})}
    \frac{(\bar{\mu}_i\tau^{-1})_{\bar{\mu}_i-1}}{\bar{\mu}_i!},  
\end{aligned}
\]
where $(a)_k$ is the Pochhammer symbol, $(a)_k = a(a+1)\cdots (a+k-1)$, 
and $z(\mu)$ and $z(\bar{\mu})$ are defined as 
\[
  z(\mu) = \prod_{i=1}^\infty i^{m_i}m_i!,\quad 
  z(\bar{\mu}) = \prod_{i=1}^\infty i^{\bar{m}_i}\bar{m}_i!
\]
in terms of the cycle type, $\mu = (1^{m_1}2^{m_2}\cdots)$ 
and $\bar{\mu} = (1^{\bar{m}_1}2^{\bar{m}_2}\cdots)$, 
of $\mu$ and $\bar{\mu}$. 

As briefly explained in Introduction, $\overline{\calM}_{g,n}$ 
is the Deligne-Mumford compactification of the moduli space 
$\calM_{g,n}$ of connected smooth complex algebraic curve $C$ 
of genus $g$ with $n$ marked points $x_1,\ldots,x_n$. 
The $i$-th $\psi$-class $\psi_i$ is the first Chern class 
$c_1(L_i)$ of a line bundle $L_i$. The fiber of $L_i$ 
at $(C,x_1,\ldots,x_n) \in \calM_{g,n}$ is the cotangent space 
$T^*_{x_i}C$ of $C$ at $x_i$.   
$\Lambda_g^\vee(u)$ is the same special linear combination 
of $1$ and the Hodge classes $\lambda_1,\dots,\lambda_g$ 
as used in (\ref{specialHI}). 
The $k$-th Hodge class $\lambda_k$ is the $k$-th 
Chern class $c_k(E_g)$ of the Hodge bundle $E_g$.  
In this sense, $\Lambda_g^\vee(u)$ is the characteristic polynomial 
in the definition of the Chern classes 
in terms of the curvature form of $E_g$. 
The fiber of $E_g$ at $(C,x_1,\ldots,x_n) \in \calM_{g,n}$ 
is the $g$-dimensional linear space 
of holomorphic $1$-forms on $C$.

\subsection{Generating functions}

Several generating functions are constructed 
from these cubic Hodge integrals.  
Firstly, the all-genus generating function $G_{\mu\bar{\mu}}(\tau)$ 
is the power series 
\beq
  G_{\mu\bar{\mu}}(\tau) 
  = \sum_{g=0}^\infty \hbar^{2g-2+l(\mu)+l(\bar{\mu})}G_{g\mu\bar{\mu}}(\tau) 
\eeq
of a new variable $\hbar$.  
Secondly, two sets of variables $\bsp = (p_k)_{k=1}^\infty$, 
$\bar{\bsp} = (\bar{p}_k)_{k=1}^\infty$ are introduced 
to make a generating function with respect 
to the partitions $\mu,\bar{\mu}$ as 
\beq
  G(\tau,\bsp,\bar{\bsp})
  = \sum_{\mu,\bar{\mu}\in\calP,(\mu,\bar{\mu})\not=(\emptyset,\emptyset)}
    G_{\mu,\bar{\mu}}(\tau)p_\mu\bar{p}_{\bar{\mu}},
\eeq
where $\calP$ denote the set of all partitions, 
and $p_\mu$ and $\bar{p}_{\bar{\mu}}$ are the monomials 
\[
  p_\mu = \prod_{i=1}^{l(\mu)} p_{\mu_i},\quad 
  \bar{p}_{\bar{\mu}} = \prod_{i=1}^{l(\bar{\mu})}\bar{p}_{\bar{\mu}_i}. 
\]
Lastly, this generating function is exponentiated to 
\beq
  G^\bullet(\tau,\bsp,\bar{\bsp}) = \exp G(\tau,\bsp,\bar{\bsp}). 
\eeq

\subsection{Schur functions}

The $\bsp$-variables originate in the so called power sums 
that relate the Schur functions $s_\mu(\bsx)$ 
in the sense of Macdonald's book \cite{Mac-book} 
to a set of polynomials $S_\mu(\bsp)$.  
$S_\mu(\bsp)$'s can be directly defined 
by the determinant formula
\[
\begin{gathered}
  S_\mu(\bsp) = \det(S_{\mu_i-i+j}(\bsp))_{i,j=1}^\infty,\\
  \sum_{m=0}^\infty S_m(\bsp)z^m 
    = \exp\left(\sum_{k=1}^\infty \frac{p_k}{k}z^k\right).
\end{gathered}
\]
The right hand side of this formula is understood 
to be an $n \times n$ determinant 
$\det(S_{\mu_i-i+j}(\bsp))_{i,j=1}^n$, $n \geq l(\mu)$, 
which is independent of $n$.  
$s_\mu(\bsx)$ is obtained from $S_\mu(\bsp)$ by substituting 
the power sums 
\[
  p_k = \sum_{i\geq 1} x_i^k,\quad k = 1,2,\ldots 
\]
of the components $x_i$ of $\bsx$. 
In the following, we need an infinite-variate version 
of the Schur functions $s_\mu(\bsx)$, $\bsx = (x_i)_{i=1}^\infty$. 
The approach to $s_\mu(\bsx)$ from $S_\mu(\bsp)$ 
is particularly convenient for that purpose.  

The skew Schur functions $s_{\mu/\nu}(\bsx)$ and $S_{\mu/\nu}(\bsp)$ 
can be treated in the same way.  The determinant formula 
of $S_\mu(\bsp)$ can be generalized to $S_{\mu/\nu}(\bsp)$ as 
\[
  S_{\mu/\nu}(\bsp) = (S_{\mu_i-\nu_j-i+j}(\bsp))_{i,j=1}^\infty. 
\]

\subsection{Combinatorial expression}

Liu, Liu and Zhou \cite{LLZ0310272} proved the following 
combinatorial expression of the generating function 
$G^\bullet(\tau,\bsp,\bar{\bsp})$. 

\begin{thm}
\beq
  G^\bullet(\tau,\bsp,\bar{\bsp}) = R^\bullet(\tau,\bsp,\bar{\bsp}),
\eeq
where $R^\bullet(\tau,\bsp,\bar{\bsp})$ is defined as 
\beq
  R^\bullet(\tau,\bsp,\bar{\bsp}) 
  = \sum_{\nu,\bar{\nu}\in\calP}
     q^{(\kappa(\nu)\tau+\kappa(\bar{\nu})\tau^{-1}))/2}
     \calW_{\nu\bar{\nu}}(q)S_\nu(\bsp)S_{\bar{\nu}}(\bar{\bsp}) 
  \label{Rb-def}
\eeq
with the combinatorial building blocks 
\[
\begin{gathered}
  \calW_{\nu\bar{\nu}}(q) = s_\nu(q^\rho)s_{\bar{\nu}}(q^{\nu+\rho}),\quad
  q = e^{\sqrt{-1}\hbar},\\
  q^\rho = (q^{-i+1/2})_{i=1}^\infty,\quad
  q^{\nu+\rho} = (q^{\nu_i-i+1/2})_{i=1}^\infty,\\
  \kappa(\nu) = \sum_{i=1}^\infty\nu_i(\nu_i-2i+1),\;
  \kappa(\bar{\nu}) = \sum_{i=1}^\infty\bar{\nu}_i(\bar{\nu}_i-2i+1). 
\end{gathered}
\]
\end{thm}

The main building block of $R^\bullet(\tau,\bsp,\bar{\bsp})$ 
is the two-leg topological vertex $\calW_{\nu\bar{\nu}}(q)$ \cite{AKMV03}. 
This is a rational function of $q^{1/2}$ that enjoys 
the non-trivial symmetries \cite{Zhou0310282} 
\beq
  \calW_{\nu\bar{\nu}}(q) 
  = \calW_{\bar{\nu}\nu}(q) 
  = (-1)^{|\nu|+|\bar{\nu}|}\calW_{\tp{\nu}\tp{\bar{\nu}}}(q^{-1}), 
  \label{W-sym}
\eeq
where $\tp{\nu}$ and $\tp{\bar{\nu}}$ stand for 
the conjugate partitions of $\nu$ and $\bar{\nu}$ 
that amount to the transpose of the associated Young diagrams. 
The symmetry under the inversion $q\mapsto q^{-1}$ 
is a consequence of the relation 
(as rational functions of $q^{1/2}$) 
\beq
  p_k(q^{\nu+\rho}) = - p_k(q^{-\tp{\nu}-\rho}) 
\eeq
between the values of the power sums 
$p_k(\bsx) = \sum_{i=1}^\infty x_i^k$ at two special points 
(cf. the proof of Lemma 1 in our previous work \cite{TN15}) 
and the identity 
\beq
  S_\mu(- \bsp) = (-1)^{|\mu|}S_{\tp{\mu}}(\bsp) 
\eeq
of the Schur functions.  

Let us mention that the last identity of the Schur functions 
can be generalized to the skew Schur functions as 
\beq
  S_{\mu/\nu}(- \bsp) = (-1)^{|\mu|+|\nu|}S_{\tp{\mu}/\tp{\nu}}(\bsp). 
\eeq
Accordingly, the inversion symmetry (\ref{W-sym}) of $\calW(q)$ 
can be extended to the three-leg topological vertex 
$C_{\lambda\mu\nu}(q)$ \cite{AKMV03} as 
\beq
  C_{\lambda\mu\nu}(q) 
  = (-1)^{|\lambda|+|\mu|+|\nu|}C_{\tp{\lambda}\tp{\mu}\tp{\nu}}(q^{-1}). 
\eeq

\section{Lift to tau function}

\subsection{Fermionic language}

The goal of this section is to convert (or, rather, lift) 
$R^\bullet(\tau,\bsp,\bar{\bsp})$ to a tau function 
of the 2D Toda hierarchy.  To this end, 
we use the language of complex fermions.  
The following (partially somewhat unusual) formulation 
of fermionic operators and Fock spaces is the same 
as our previous work \cite{Takasaki13a,Takasaki13b,Takasaki14,TN15}. 

Let $\psi_n$ and $\psi^*_n$, $n \in \ZZ$, denote the Fourier modes
\footnote{Do not confuse them with the $\psi$-classes 
on $\bar{\calM}_{g,n}$.  Moreover, as opposed to the usual formulation, 
we label these operators with integers rather than half integers.}
of the 2D free fermion fields 
\[
  \psi(z) = \sum_{n\in\ZZ}\psi_nz^{-n},\quad 
  \psi^*(z) = \sum_{n\in\ZZ}\psi^*_nz^{-n-1} 
\]
that satisfy the anti-commutation relations 
\[
  \psi_m\psi^*_n + \psi^*_n\psi_m = \delta_{m+n,0}, \quad
  \psi_m\psi_n + \psi_n\psi_m = 0, \quad 
   \psi^*_m\psi^*_n + \psi^*_n\psi^*_m = 0. 
\]
The (bra- and ket-) Fock spaces are decomposed 
to the charge-$s$ sectors for $s \in \ZZ$.  
The charge-$s$ sectors are spanned by the vectors 
\[
\begin{aligned}
  \langle \mu,s| &= \langle -\infty|\cdots\psi^*_{\mu_i-i+1+s}
    \cdots\psi^*_{\mu_2-1+s}\psi^*_{\mu_1+s},\\
  |\mu,s\rangle &= \psi_{-\mu_1-s}\psi_{-\mu_2+1-s}\cdots
    \psi_{-\mu_i+i-1-s}\cdots|-\infty\rangle 
\end{aligned}
\]
labelled by partitions $\mu\in\calP$. Their pairing 
is defined as 
\[
  \langle\mu,r|\nu,s\rangle = \delta_{rs}\delta_{\mu\nu}. 
\]
The vectors $\langle\emptyset,s|$ and $|\emptyset,s\rangle$
are called the ground states of the charge-$s$ sector, 
and abbreviated as $\langle s|$ and $|s\rangle$. 
In particular, $\langle 0|$ and $|0\rangle$ represent 
the vacuum states of the whole fermion system. 
Moreover, let $\langle\mu|$ and $|\mu\rangle$ 
denote the vectors $\langle\mu,0|$ and $|\mu,0\rangle$ 
in the charge-$0$ sector.  

Let $L_0,K$ and $J_m$, $m \in \ZZ$, denote the special 
fermion bilinears 
\[
\begin{gathered}
    L_0 = \sum_{n\in\ZZ}n{:}\psi_{-n}\psi^*_n{:},\quad
    K = \sum_{n\in\ZZ}(n-1/2)^2{:}\psi_{-n}\psi^*_n{:},\\
    J_m = \sum_{n\in\ZZ}{:}\psi_{-n}\psi^*_{n+m}{:}, \quad 
    s \in \ZZ, 
\end{gathered}
\]
where ${:}\psi_{-m}\psi^*_n{:}$'s are the normal ordered product: 
\[
\begin{gathered}
  {:}\psi_{-m}\psi^*_n{:} 
  = \psi_{-m}\psi^*_n - \langle 0|\psi_{-m}\psi^*_n|0\rangle,\\
  \langle 0|\psi_{-m}\psi^*_n|0\rangle 
  = \begin{cases}
    1 & \text{if $m = n \leq 0$},\\
    0 & \text{otherwise}. 
    \end{cases}
\end{gathered}
\]
$J_m$'s are used to construct the vertex operators \cite{ORV03,BY08} 
\[
\begin{gathered}
  \Gamma_{\pm}(z) 
  = \exp\left(\sum_{k=1}^\infty\frac{z^k}{k}J_{\pm k}\right),\\
  \Gamma'_{\pm}(z) 
  = \exp\left(- \sum_{k=1}^\infty\frac{(-z)^k}{k}J_{\pm k}\right) 
\end{gathered}
\]
and the multi-variate extensions 
\[
  \Gamma_{\pm}(\bsx) = \prod_{i=1}^\infty\Gamma_{\pm}(x_i),\quad 
  \Gamma'_{\pm}(\bsx) = \prod_{i\ge 1}\Gamma'_{\pm}(x_i). 
\]

The action of these operators preserves the charge. 
The matrix elements in the charge-$s$ sector 
take the following form: 
\begin{gather}
  \langle\mu,s|L_0|\nu,s\rangle = \delta_{\mu\nu}(|\mu| + s(s+1)/2),
  \label{<L0>}\\
  \langle\mu,s|K|\nu,s\rangle  
    = \delta_{\mu\nu}(\kappa(\mu) + 2s|\mu| + (4s^3 - s)/12),
  \label{<K>}\\
  \langle\mu,s|\Gamma_{-}(\bsx)|\nu,s\rangle
  = \langle\nu,s|\Gamma_{+}(\bsx)|\mu,s\rangle
  = s_{\mu/\nu}(\bsx),
  \label{<Gamma>}\\
  \langle\lambda,s|\Gamma'_{-}(\bsx)|\mu,s\rangle 
  = \langle\mu,s|\Gamma'_{+}(\bsx)|\lambda,s\rangle
  = s_{\tp{\lambda}/\tp{\mu}}(\bsx). 
  \label{<Gamma'>}
\end{gather}

The matrix elements of the vertex operators are thus 
independent of $s$, and yield the fermionic expression 
\beq
  s_{\mu/\nu}(\bsx) = \langle\mu|\Gamma_{-}(\bsx)|\nu\rangle
  = \langle\nu|\Gamma_{+}(\bsx)|\mu\rangle
  \label{S(x)=<>}
\eeq
of the skew Schur functions of the $\bsx$-variables. 
(Recall that $\langle\mu|$, etc. are abbreviations 
of $\langle\mu,0|$, etc.)  This expression can be derived 
from the fermionic expression \cite{MJD-book} 
\beq
  S_{\mu/\nu}(\bsp) 
  = \langle\mu|\exp\left(\sum_{k=1}^\infty\frac{p_k}{k}J_{-k}\right)|\nu\rangle 
  = \langle\nu|\exp\left(\sum_{k=1}^\infty\frac{p_k}{k}J_k\right)|\mu\rangle 
  \label{S(p)=<>}
\eeq
of the skew Schur functions of the $\bsp$-variables.

\subsection{Fermionic expression of $\calW_{\nu\bar{\nu}}(q)$}

The two-leg topological vertex has yet another expression 
\cite{AKMV03,Zhou0310282}: 
\beq
  \calW_{\nu\bar{\nu}}(q) 
  = q^{(\kappa(\nu)+\kappa(\bar{\nu}))/2}
    \sum_{\eta\in\calP}s_{\tp{\nu}/\eta}(q^\rho)s_{\tp{\bar{\nu}}/\eta}(q^\rho). 
  \label{W(q)3rd}
\eeq
The right hand side of expression is actually a finite sum 
over all $\eta$'s with $\eta\subseteq\tp{\nu}$ 
and $\eta\subseteq\tp{\bar{\nu}}$, where $\subseteq$ stands 
for the inclusion relation of the associated Young diagrams.  
We can translate this expression to the language of fermions 
as follows: 

\begin{lemma}
\beq
  \calW_{\nu\bar{\nu}}(q) 
  = \langle\tp{\nu}|q^{-K/2}\Gamma_{-}(q^\rho)\Gamma_{+}(q^\rho)q^{-K/2}
    |\tp{\bar{\nu}}\rangle.
  \label{W(q)=<>1st}
\eeq
Here (and in the rest of this paper) the $q$-exponential $q^A$ 
of an operator $A$ stands for $\exp(A\log q)$.  
\end{lemma}

\begin{proof}
By the fermionic formula (\ref{S(x)=<>}) 
of the skew Schur functions and the partition of unity 
$1 = \sum_{\eta\in\calP}|\eta\rangle\langle\eta|$ 
in the charge-$0$ sector, we can express the sum 
on the right hand side of (\ref{W(q)3rd}) as 
\[
\begin{aligned}
  \sum_{\eta\in\calP}s_{\tp{\nu}/\eta}(q^\rho)s_{\tp{\bar{\nu}}/\eta}(q^\rho) 
  &= \sum_{\eta\in\calP}\langle\tp{\nu}|\Gamma_{-}(q^\rho)|\eta\rangle
    \langle\eta|\Gamma_{+}(q^\rho)|\tp{\bar{\nu}}\rangle \\
  &= \langle\tp{\nu}|\Gamma_{-}(q^\rho)\Gamma_{+}(q^\rho)
     |\tp{\bar{\nu}}\rangle. 
\end{aligned}
\]
To move the remaining $q$-factors inside 
$\langle\tp{\nu}|\cdots|\tp{\bar{\nu}}\rangle$, 
we use the relation 
\[
  \kappa(\tp{\nu}) = - \kappa(\nu)
\]
and the formula (\ref{<K>}) of the matrix elements of $K$ as 
\[
\begin{aligned}
   q^{\kappa(\nu)/2+\kappa(\bar{\nu})/2}
    \langle\tp{\nu}|\cdots|\tp{\bar{\nu}}\rangle
  &= q^{-\kappa(\tp{\nu})/2}\langle\tp{\nu}|\cdots
      |\tp{\bar{\nu}}\rangle q^{-\kappa(\tp{\bar{\nu}})/2}\\
  &= \langle\tp{\nu}|q^{-K/2}\cdots q^{-K/2}|\tp{\bar{\nu}}\rangle. 
\end{aligned}
\]
\end{proof}

(\ref{W(q)=<>1st}) is an expression of the rational function 
$\calW_{\nu\bar{\nu}}(q)$ in the region $|q| > 1$.  
The vertex operators in this expression can be computed as 
\[
  \Gamma_\pm(q^\rho) 
  = \exp\left(\sum_{k,i=1}^\infty\frac{q^{-(i-1/2)k}}{k}J_{\pm k}\right) 
  = \exp\left(\sum_{k=1}^\infty\frac{q^{-k/2}}{k(1-q^{-k})}J_{\pm k}\right), 
\]
and this computation is valid only in the region $|q| > 1$.  

If we can now start from the last operator 
(whose matrix elements are rational functions of $q^{1/2}$) 
and rewrite it as 
\[
    \exp\left(\sum_{k=1}^\infty\frac{q^{-k/2}}{k(1-q^{-k})}J_{\pm k}\right) 
  =  \exp\left(- \sum_{k=1}^\infty\frac{q^{k/2}}{k(1-q^k)}J_{\pm k}\right),
\]
we can proceed in an opposite direction as 
\[
  \exp\left(- \sum_{k=1}^\infty\frac{q^{k/2}}{k(1-q^k)}J_{\pm k}\right) 
    = \exp\left(- \sum_{k,i=1}^\infty\frac{q^{(i-1/2)k}}{k}J_{\pm k}\right) 
  = \Gamma'_\pm(-q^{-\rho}). 
\]
Note that this computation is valid in the region $|q| < 1$. 

These considerations show that the operator $\Gamma_\pm(q^\rho)$ 
in the region $|q| > 1$ and the operator $\Gamma'_\pm(- q^{-\rho})$ 
are {\it analytic continuation} of each other.  
We are thus led to the following expression of $\calW_{\nu\bar{\nu}}(q)$ 
that is valid in the region $|q| < 1$: 
\beq
  \calW_{\nu\bar{\nu}}(q) 
  = \langle\tp{\nu}|q^{-K/2}\Gamma'_{-}(-q^{-\rho})
    \Gamma'_{+}(-q^{-\rho})q^{-K/2}|\tp{\bar{\nu}}\rangle. 
  \label{W(q)=<>2nd}
\eeq

Actually, this is an intermediate stage. 
We rewrite it further as follows. 

\begin{lemma}
\beq
  \calW_{\nu\bar{\nu}}(q) 
  = (-1)^{|\nu|+|\bar{\nu}|}\langle\nu|q^{K/2}\Gamma_{-}(q^{-\rho})
    \Gamma_{+}(q^{-\rho})q^{K/2}|\bar{\nu}\rangle. 
  \label{W(q)=<>3rd}
\eeq
\end{lemma}

\begin{proof}
Let us rewrite the right hand side of (\ref{W(q)=<>2nd}) as 
\[
\begin{aligned}
  & \langle\tp{\nu}|q^{-K/2}\Gamma'_{-}(-q^{-\rho})
    \Gamma'_{+}(-q^{-\rho})q^{-K/2}|\tp{\bar{\nu}}\rangle\\
  &= q^{(\kappa(\nu)+\kappa(\bar{\nu}))/2}
     \langle\tp{\nu}|\Gamma'_{-}(-q^{-\rho})\Gamma'_{+}(-q^{-\rho})
     |\tp{\bar{\nu}}\rangle.
\end{aligned}
\]
Since $L_0$ and $J_k$'s satisfy the commutation relations
\[
  [L_0,J_k] = -kJ_k, 
\]
the negative sign in front of $q^{-\rho}$ can be eliminated 
by the adjoint action of $(-1)^{L_0}$ as 
\[
  \Gamma'_\pm(-q^{-\rho}) = (-1)^{L_0}\Gamma'_\pm(q^{-\rho})(-1)^{L_0}. 
\]
Consequently,  
\[
  \langle\tp{\nu}|\Gamma'_{-}(-q^{-\rho})\Gamma'_{+}(-q^{-\rho})
  |\tp{\bar{\nu}}\rangle 
  = (-1)^{|\nu|+|\bar{\nu}|}\langle\tp{\nu}|\Gamma'_{-}(q^{-\rho})
    \Gamma'_{+}(q^{-\rho})|\tp{\bar{\nu}}\rangle. 
\]
Since the matrix elements (\ref{<Gamma>}) and (\ref{<Gamma'>}) 
of the two types of vertex operators correspond to each other 
by transposing the partitions, we have the identity 
\[
  \langle\tp{\nu}|\Gamma'_{-}(q^{-\rho})\Gamma'_{+}(q^{-\rho})
  |\tp{\bar{\nu}}\rangle 
  = \langle\nu|\Gamma_{-}(q^{-\rho})\Gamma_{+}(q^{-\rho})|\bar{\nu}\rangle.
\]
Collecting these formulae yields (\ref{W(q)=<>3rd})
\end{proof}

We use the last fermionic expression (\ref{W(q)=<>3rd}) 
to convert $R^\bullet(\tau,\bsp,\bar{\bsp})$ 
to a tau function of the 2D Toda hierarchy.  
As mentioned above, this expression itself is valid 
in the region $|q| < 1$.  It is easy to see that 
the first expression (\ref{W(q)=<>1st}) is connected 
with this expression by the inversion $q \to q^{-1}$: 
\beq
  \calW_{\nu\bar{\nu}}(q) 
  = (-1)^{|\nu|+|\bar{\nu}}W_{\tp{\nu}\tp{\bar{\nu}}}(q^{-1}). 
\eeq
This is exactly the inversion relation (\ref{W-sym}) 
mentioned in the end of the previous section. 
We have derived it from a slightly different route.

\subsection{Lifting $R^\bullet(\tau,\bsp,\bar{\bsp})$ to tau function}

All building blocks of the definition 
(\ref{Rb-def}) of $R^\bullet(\tau,\bsp,\bar{\bsp})$ 
are now translated to the language of fermions. 
This leads to the following fermionic expression 
of $R^\bullet(\tau,\bsp,\bar{\bsp})$. 

\begin{thm}\label{Rb=<>-thm}
\beq
  R^\bullet(\tau,\bsp,\bar{\bsp}) 
  = \langle 0|\exp\left(\sum_{k=1}^\infty\frac{(-1)^kp_k}{k}J_k\right)
    h\exp\left(\sum_{k=1}^\infty\frac{(-1)^k\bar{p}_k}{k}J_{-k}\right)|0\rangle,
  \label{Rb=<>}
\eeq
where 
\beq
  h = q^{(\tau+1)K/2}\Gamma_{-}(q^{-\rho})\Gamma_{+}(q^{-\rho})q^{(\tau^{-1}+1)K/2}. 
  \label{h-def}
\eeq
\end{thm}

\begin{proof}
Let us apply the fermionic formulae (\ref{<K>}), (\ref{S(p)=<>}) 
and (\ref{W(q)=<>3rd}) to the building blocks of (\ref{Rb-def}). 
We want to achieve the summation over $\nu,\bar{\nu} \in \calP$ 
with the aid of the partition of unity 
\[
  1 = \sum_{\nu\in\calP}|\nu\rangle\langle\nu| 
    = \sum_{\bar{\nu}\in\calP}|\bar{\nu}\rangle\langle\bar{\nu}| 
\]
in the charge-$0$ sector.  An obstacle is the sign factor 
$(-1)^{|\nu|+|\bar{\nu}|}$.  This factor can be absorbed 
by the Schur functions as 
\[
  (-1)^{|\nu|+|\bar{\nu}|}S_\nu(\bsp)S_{\bar{\bar{\nu}}}(\bar{\bsp}) 
  = S_\nu(\ldots,(-1)^kp_k,\ldots)
    S_{\bar{\nu}}(\ldots,(-1)^k\bar{p}_k,\ldots). 
\]
We can thus use the partition of unity to obtain (\ref{Rb=<>}). 
\end{proof}

(\ref{Rb=<>}) is very close to a general fermionic expression 
of tau functions of the 2D Toda hierarchy \cite{Takebe91,Zabrodin13} 
(see also ref. \cite{Takasaki18a}, Section 3.3). 
A tau function $\calT(s,\bst,\bar{\bst})$ 
of the lattice coordinate $s$ and the time variables 
$\bst = (t_k)_{k=1}^\infty$ and $\bar{\bst} = (\bar{t}_k)_{k=1}^\infty$ 
can be obtained by replacing 
\[
\begin{gathered}
  \frac{(-1)^kp_k}{k} \to t_k, \quad 
  \frac{(-1)^p\bar{p}_k}{k} \to - \bar{t}_k,\\
  \langle 0| \to \langle s|,\quad 
  |0\rangle \to |s\rangle,\quad s \in \ZZ. 
\end{gathered}
\]
as 
\beq
  \calT(s,\bst,\bar{\bst}) 
  = \langle s|\exp\left(\sum_{k=1}^\infty t_kJ_k\right)
    h\exp\left(- \sum_{k=1}^\infty\bar{t}_kJ_{-k}\right)|s\rangle. 
  \label{Tau-def}
\eeq
This is the tau function constructed by Zhou \cite{Zhou0310408}.

\subsection{$s$-dependence of $\calT(s,\bst,\bar{\bst})$}

By twice inserting the partition of unity as we have done 
in the proof of Theorem \ref{Rb=<>-thm}, we can expand 
$\calT(s,\bst,\bar{\bst})$ into a double sum over partitions: 
\beq
\begin{aligned}
  \calT(s,\bst,\bar{\bst}) 
  &= \sum_{\nu,\bar{\nu}\in\calP}
     q^{(\tau+1)(\kappa(\nu)/2+s|\nu|+(4s^3-s)/24)}\\
  &\qquad\mbox{}\times
     q^{(\tau^{-1}+1)(\kappa(\bar{\nu})/2+s|\bar{\nu}|+(4s^3-s)/24)}\\
  &\qquad\mbox{}\times
     \langle\nu|\Gamma_{-}(q^{-\rho})\Gamma_{+}(q^{-\rho})|\bar{\nu}\rangle
     \rmS_\nu(\bst)\rmS_{\bar\nu}(\bar{\bst}).  
\end{aligned}
\label{Tau-dsum}
\eeq
$\rmS_\nu(\bst)$ and $\rmS_{\bar{\nu}}(\bar{\bst})$ 
are the Schur functions $S_\nu(\bsp)$ and $S_{\bar{\nu}}(\bar{\bsp})$ 
regarded as functions of $\bst$ and $\bar{\bst}$ 
by the relation 
\[
  p_k = kt_k, \quad \bar{p}_k = k\bar{t}_k,\quad k = 1,2,\ldots. 
\]
The $s$-dependent factors come from the matrix elements 
(\ref{<K>}) of $K$.

Although $s$ is originally a lattice coordinate, 
hence a discrete variable, (\ref{Tau-dsum}) hints 
that $s$ may be thought of as a continuous variable. 
The following fact shows that this point of view 
is meaningful enough. 

\begin{thm}\label{Tau(s+c)-thm}
For any constant $c$, 
the function $\calT(s+c,\bst,\bar{\bst})$ restricted 
to $s \in \ZZ$ is a tau function of the 2D Toda hierarchy. 
\end{thm}

\begin{proof}
When $s$ is shifted to $s + c$, the exponents of 
the two exponential factors in (\ref{Tau-dsum}) vary as 
\[
\begin{aligned}
  &\kappa(\nu)/2+s|\nu|+(4s^3-s)/24\\
  &\to \kappa(\nu)/2+(s+c)|\nu|+(4(s+c)^3-(s+c))/24) \\
  &= \langle\nu,s|(K/2 + cL_0 + (c^2-c)J_0/2)|\nu,s\rangle 
     + (4c^3 - c)/24 
\end{aligned}
\]
and 
\[
\begin{aligned}
  &\kappa(\bar{\nu})/2+s|\bar{\nu}|+(4s^3-s)/24\\
  &\to \kappa(\bar{\nu})/2+(s+c)|\bar{\nu}|+(4(s+c)^3-(s+c)/24) \\
  &= \langle\bar{\nu},s|(K/2 + cL_0 + (c^2-c)J_0/2)|\bar{\nu},s\rangle 
     + (4c^3 - c)/24. 
\end{aligned}
\]
This implies that the shifted tau function $\calT(s+c,\bst,\bar{\bst})$ 
can be expressed in the fermionic form (\ref{Tau-def}) 
with $h$ replaced by 
\[
\begin{aligned}
  h(c) &= q^{(\tau+\tau^{-1}+2)(4c^3-c)/24)}q^{(\tau+1)(K/2+cL_0+(c^2-c)J_0/2)}\\
       &\quad\mbox{}\times 
        \Gamma_{-}(q^{-\rho})\Gamma_{+}(q^{-\rho})
        q^{(\tau^{-1}+1)(K/2+cL_0+(c^2-c)J_0/2)}. 
\end{aligned}
\]
\end{proof}

The tau function $\calT(s,\bst,\bar{\bst})$ thus yields 
a solution of the Toda hierarchy on any integral sublattice 
$\ZZ + c \subset \RR$.  In the following section, 
we consider this solution in the Lax formalism.

\section{Perspectives in Lax formalism}

\subsection{Fractional powers of Lax operators}

Let $\Lambda$ denote the shift operator in the variable $s$: 
\[
  \Lambda = e^{\rd_s}, \quad \Lambda^nf(s) = f(s+n). 
\]
The Lax formalism of the 2D Toda hierarchy 
uses two (so to speak, {\it pseudo-difference}) 
operators $L,\bar{L}$ of the form
\[
\begin{gathered}
  L = \Lambda + \sum_{n=1}^\infty u_n\Lambda^{1-n},\quad 
  \bar{L}^{-1} = \sum_{n=0}^\infty\bar{u}_n\Lambda^{n-1},\\
  u_n = u_n(s,\bst,\bar{\bst}), \quad 
  \bar{u}_n = \bar{u}_n(s,\bst,\bar{\bst}), 
\end{gathered}
\]
that satisfy the Lax equations 
\[
\begin{gathered}
  \frac{\rd L}{\rd t_k} = [B_k,L],\quad 
  \frac{\rd L}{\rd\bar{t}_k} = [\bar{B}_k,L],\\
  \frac{\rd \bar{L}}{\rd t_k} = [B_k,\bar{L}],\quad 
  \frac{\rd \bar{L}}{\rd\bar{t}_k} = [\bar{B}_k,\bar{L}],\\
  B_k = (L^k)_{\geq 0},\quad \bar{B}_k = (\bar{L}^{-k})_{<0}, 
\end{gathered}
\]
where $(\quad)_{\ge 0}$ and $(\quad)_{<0}$ denote 
the projection onto the non-negative and negative power parts 
of difference operators: 
\[
  \left(\sum_{n\in\ZZ}a_n\Lambda^n\right)_{\ge 0} 
    = \sum_{n\geq 0}a_n\Lambda^n,\quad 
  \left(\sum_{n\in\ZZ}a_n\Lambda^n\right)_{<0} 
    = \sum_{n<0}a_n\Lambda^n. 
\]

By the standard procedure \cite{UT84,TT95,Takasaki18a}, 
the tau function $\calT(s,\bst,\bar{\bst})$ yields 
the Lax operators via two dressing operators 
\[
    W = 1 + \sum_{n=1}^\infty w_ne^{-n\rd_s},\quad 
  \bar{W} = \sum_{n=0}^\infty\bar{w}_ne^{n\rd_s},\quad 
  \bar{w}_0 \not= 0.
\]
The coefficients are obtained from the Laurent expansion 
\[
\begin{gathered}
  \frac{\calT(s-1,\bst-[z^{-1}],\bar{\bst})}{\calT(s-1,\bst,\bar{\bst})}
    = 1 + \sum_{n=1}^\infty w_nz^{-n},\\
  \frac{\calT(s,\bst,\bar{\bst}-[z])}{\calT(s-1,\bst,\bar{\bst})}
    = \sum_{n=0}^\infty\bar{w}_nz^n,\\
  [z] = \left(z,z^2/2,\cdots,z^k/k,\cdots\right),   
\end{gathered}
\]
of quotients of two shifted tau functions 
at $z = \infty$ and $z = 0$, respectively.  
The Lax operators are thereby expressed as 
\[
  L = W\Lambda W^{-1},\quad 
  \bar{L}^{-1} = \bar{W}\Lambda^{-1}\bar{W}^{-1}. 
\]

Since $s$ is now interpreted to be a continuous variable, 
we can define the logarithm and the fractional powers 
of $L$ and $\bar{L}$ as 
\beq
\begin{gathered}
  \log L = W\log\Lambda W^{-1},\quad 
  \log\bar{L} = \bar{W}\log\Lambda\bar{W}^{-1},\\
  L^\alpha = W\Lambda^\alpha W^{-1},\quad
  \bar{L}^\alpha = \bar{W}\Lambda^\alpha\bar{W}^{-1}.  
\end{gathered}
\eeq
Note that $\log\Lambda$ and $\Lambda^\alpha$ are 
differential and shift operators in the variable $s$: 
\[
  \log\Lambda = \rd_s,\quad 
  \Lambda^{\alpha} = e^{\alpha\rd_s}. 
\]
Consequently, $\log L$, $\log\bar{L}$, $L^\alpha$ 
and $\bar{L}^\alpha$ can be expressed as 
\[
\begin{gathered}
  \log L = \rd_s - \frac{\rd W}{\rd s}W^{-1},\quad 
  \log\bar{L} = \rd_s - \frac{\rd\bar{W}}{\rd s}\bar{W}^{-1}, \\
  L^\alpha = W\cdot W^{-1}|_{s\to s+\alpha}\cdot e^{\alpha\rd_s} 
  = (1 + p_1\Lambda^{-1} + \cdots)e^{\alpha\rd_s},\\
  \bar{L}^\alpha = \bar{W}\cdot\bar{W}^{-1}|_{s\to s+\alpha}\cdot e^{\alpha\rd_S} 
  = (\bar{p}_0 + \bar{p}_1\Lambda + \cdots)e^{\alpha\rd_s}, 
\end{gathered}
\]
where $p_1,p_2,\ldots$ and $\bar{p}_0,\bar{p}_1,\ldots$ 
are determined by the coefficients of $W$ and $\bar{W}$, 
e.g., 
\beq
  p_1 = w_1(s,\bst,\bar{\bst}) - w_1(s+\alpha,\bst,\bar{\bst}),\quad
  \bar{p}_0 = \frac{\bar{w}_0(s,\bst,\bar{\bst})}
              {\bar{w}_0(s+\alpha,\bst,\bar{\bst})}. 
\eeq
Calculus of this kind of fractional difference operators 
is discussed in the paper of Liu, Zhang and Zhou \cite{LZZ17}.

\subsection{Algebraic relations of Lax operators} 

We can now state the key theorem: 

\begin{thm}\label{key-thm}
The Lax operators obtained from the tau function 
$\calT(s,\bst,\bar{\bst})$ of (\ref{Tau-def}) 
satisfy the algebraic relation 
\beq
  L^{1/(\tau+1)} = - \bar{L}^{-1/(\tau^{-1}+1)}. 
  \label{LLbar-rel}
\eeq
\end{thm}

Implications of this algebraic relation are discussed 
in the next section. The rest of this section is devoted 
to proving this theorem.  Let us note here the following 
immediate consequence. 

\begin{cor}\label{key-cor}
There is a function $u = u(s,\bst,\bar{\bst})$ such that 
\beq
  L^{1/(\tau+1)} = - \bar{L}^{-1/(\tau^{-1}+1)} 
  = (1 - u\Lambda^{-1})\Lambda^{1/(\tau+1)}
  \label{LLbar-u}
\eeq
\end{cor}

\begin{proof}
The both sides of (\ref{LLbar-rel}) are fractional 
difference operators of the following form: 
\[
\begin{aligned}
  L^{1/(\tau+1)} 
  &= (1 + p_1\Lambda^{-1} + p_2\Lambda^{-2} + \cdots)\Lambda^{1/(\tau+1)}\\
  &= \Lambda^{1/(\tau+1)} + p_1\Lambda^{-\tau/(\tau+1)} + \cdots,\\
  - \bar{L}^{-1/(\tau^{-1}+1)} 
  &= - (\bar{p}_0 + \bar{p}_1\Lambda + \bar{p}_2\Lambda^2+ \cdots)
     \Lambda^{-\tau/(\tau+1)}\\
  &= - \bar{p}_0\Lambda^{-\tau/(\tau+1)} - \bar{p}_1\Lambda^{1/(\tau+1)} + \cdots. 
\end{aligned}
\]
Therefore only two terms survive: 
\[
\begin{gathered}
  L^{1/(\tau+1)} = - \bar{L}^{-1/(\tau^{-1}+1)} 
  = \Lambda^{1/(\tau+1)} + p_1\Lambda^{-\tau/(\tau+1)},\\
  \bar{p}_1 = -1, \quad \bar{p_0} = - p_1. 
\end{gathered}
\]
\end{proof}

\subsection{Factorization problem}

The proof of Theorem \ref{key-thm} borrows a main idea 
from our previous work \cite{Takasaki13a,Takasaki13b,Takasaki14}. 
The idea is to use the factorization problem 
\beq
  \exp\left(\sum_{k=1}^\infty t_k\Lambda^k\right)
  U\exp\left(- \sum_{k=1}^\infty\bar{t}_k\Lambda^{-k}\right)
  = W^{-1}\bar{W} 
  \label{factor}
\eeq
at a particular point of the $(\bst,\bar{\bst})$-space 
where the factors $W,\bar{W}$ can be obtained explicitly. 

In the usual setting \cite{UT84,TT95,Takasaki84}, 
(\ref{factor}) is an equation for $\ZZ\times\ZZ$ matrices.  
Namely, given an invertible (typically, exponential) matrix $U$, 
the problem is to find a lower-triangular matrix $W$ 
and an upper-triangular matrix $\bar{W}$ satisfying 
(\ref{factor}).  Arbitrariness of $W$ and $\bar{W}$ 
(i.e., gauge freedom $W \to DW$, $\bar{W} \to D\bar{W}$ 
by a diagonal matrix $D$) disappears 
under the normalization condition 
that all diagonal elements of $W$ are equal to $1$. 
The Lax operators defined by $W$ and $\bar{W}$ as 
\[
  L = W\Lambda W^{-1}, \quad 
  \bar{L}^{-1} = \bar{W}\Lambda\bar{W}^{-1} 
\]
give a solution of the 2D Toda hierarchy.  

Moreover, this $\ZZ\times\ZZ$ matrix formalism 
is directly related to the fermionic expression 
of tau functions like (\ref{Tau-def}) \cite{Takebe91,Zabrodin13}.  
If $h$ takes the exponential form $h = e^{\hat{A}}$ 
with a fermion bilinear 
\[
  \hat{A} = \sum_{i,j\in\ZZ}a_{ij}{:}\psi_{-i}\psi^*_j{:}, 
\]
$U$ becomes the exponential matrix $e^A$ with 
\[
  A = (a_{ij})_{i,j\in\ZZ}. 
\]
If $h$ is a product of such exponential operators, 
$U$ is a product of the associated exponential matrices. 
In this sense, $K$, $L_0$ and $J_k$ correspond to 
diagonal and shift matrices as 
\footnote{To simplify the notations, we use $\Lambda^k$ 
for both the matrix and the shift operator.}
\beq
  K \leftrightarrow (\Delta - 1/2)^2,\quad 
  L_0 \leftrightarrow \Delta,\quad 
  J_k \leftrightarrow \Lambda^k, 
\eeq
where 
\[
  1 = (\delta_{ij})_{i,j\in\ZZ},\quad 
  \Delta = (i\delta_{ij})_{i,j\in\ZZ},\quad 
  \Lambda^k = (\delta_{i+k,j})_{i,j\in\ZZ}. 
\]
The single-variate vertex operators 
\[
  \Gamma_\pm(z) 
  = \exp\left(\sum_{k=1}^\infty\frac{z^k}{k}J_{\pm k}\right) 
\]
amount to the matrix 
\[
  \exp\left(\sum_{k=1}^\infty\frac{z^k}{k}\Lambda^{\pm k}\right) 
  = \exp\left(- \log(1 - z\Lambda^{\pm 1})\right)
  = (1 - z\Lambda^{\pm 1})^{-1}, 
\]
hence 
\[
  \Gamma_\pm(q^{-\rho}) = \prod_{i=1}^\infty\Gamma_\pm(q^{i-1/2}) 
  \longleftrightarrow 
  \prod_{i=1}^\infty(1 - q^{i-1/2}\Lambda^{\pm 1})^{-1}. 
\]
Thus the operator $h$ of (\ref{h-def})corresponds 
to the following matrix: 
\beq
\begin{aligned}
  U &= q^{(\tau+1)(\Delta-1/2)^2/2}\cdot
      \prod_{k=1}^\infty(1 - q^{i-1/2}\Lambda^{-1})^{-1}\\
    &\quad\mbox{}\times 
      \prod_{k=1}^\infty(1 - q^{i-1/2}\Lambda)^{-1}\cdot
      q^{(\tau^{-1}+1)(\Delta-1/2)^2/2}. 
\end{aligned}
\label{U-mat}
\eeq
      
In the present setting, $s$ is considered to be 
a continuous variable.  Therefore we interpret 
(\ref{factor}) as equations for operators 
on the continuous space $\RR$.  Accordingly, 
$K$, $L_0$ and $\Lambda^k$ now correspond to 
multiplication and difference operators as 
\beq
  K \leftrightarrow (s- 1/2)^2,\quad 
  L_0 \leftrightarrow s,\quad 
  J_k \leftrightarrow \Lambda^k = e^{k\rd_s}. 
\eeq
Thus the matrix $U$ of (\ref{U-mat}) is replaced by the operator 
\beq
\begin{aligned}
  U &= q^{(\tau+1)(s-1/2)^2/2}\cdot
      \prod_{k=1}^\infty(1 - q^{i-1/2}\Lambda^{-1})^{-1}\\
    &\quad\mbox{}\times 
      \prod_{k=1}^\infty(1 - q^{i-1/2}\Lambda)^{-1}\cdot
      q^{(\tau^{-1}+1)(s-1/2)^2/2}. 
\end{aligned}
\label{U-op}
\eeq

Note that this interpretation is consistent with 
the computation in the proof of Theorem \ref{Tau(s+c)-thm}.  
We have seen therein that the operator $K$ in $h$ 
turns into $K + cL_0 + c(c-1)J_0/2$ as $s$ is shifted to $s + c$. 
This exactly corresponds to the variation 
\[
  (s - 1/2)^2 \to (s + c - 1/2)^2 = (s - 1/2)^2 + cs + c(c-1)/2 
\]
of the associated multiplication operator.

\subsection{Initial values of operators}

We use the factorization problem (\ref{factor}) 
to compute the initial values 
of $L^{1/(\tau+1)}$ and $\bar{L}^{-1/(\tau^{-1}+1)}$ 
at $\bst = \bar{\bst} = \bszero$.  
The first step toward this end is to find 
the initial values of the dressing operators. 

When $\bst = \bar{\bst} = \bszero$, 
the factorization problem takes the simpler form 
\beq
  U = W_{\bszero}^{-1}\bar{W}_{\bszero}, 
\label{factor0}
\eeq
where 
\[
   W_{\bszero} = W|_{\bst=\bar{\bst}=\bszero},\quad 
   \bar{W}_{\bszero} = \bar{W}_{\bst=\bar{\bst}=\bszero}. 
\]
$W_{\bszero}$ and $\bar{W}_{\bszero}$ are the initial values 
of the dressing operators, and thus characterized 
by the simpler factorization problem (\ref{factor0}). 

Unlike the original factorization problem (\ref{factor}), 
this factorization problem can be solved easily, 
because the operator $U$ of (\ref{U-op}) is 
already factorized in an almost final form.  
The only thing to do is to adjust the factors 
by gauge freedom so that the leading coefficient 
of the $\Lambda^{-1}$-expansion of the first factor 
is equal to $1$.  We thus obtain the following expression 
of $W_{\bszero}$ and $\bar{W}_{\bszero}$: 
\begin{gather}
  W_{\bszero} 
  = q^{(\tau+1)(s-1/2)^2/2}\cdot\prod_{i=1}^\infty(1 - q^{i-1/2}\Lambda^{-1})
    \cdot q^{-(\tau+1)(s-1/2)^2/2},
  \label{W0}\\
  \bar{W}_{\bszero} 
  = q^{(\tau+1)(s-1/2)^2/2}\cdot\prod_{i=1}^\infty(1 - q^{i-1/2}\Lambda)^{-1}
    \cdot q^{(\tau^{-1}+1)(s-1/2)^2/2}.
  \label{Wbar0} 
\end{gather}

This enables us to compute the initial values 
\[
  L_{\bszero}^{1/(\tau+1)} 
    = L^{1/(\tau+1)}|_{\bst=\bar{\bst}=\bszero}, \quad 
  \bar{L}^{-1/(\tau^{-1}+1)}_{\bszero} 
    = \bar{L}^{-1/(\tau^{-1}+1)}|_{\bst=\bar{\bst}=\bszero}
\]
of $L^{1/(\tau+1)}$ and $\bar{L}^{-1/(\tau^{-1}+1)}$ 
from $W_{\bszero}$ and $\bar{W}_{\bszero}$ as 
\[
  L_{\bszero}^{1/(\tau+1)} 
    = W_{\bszero}\Lambda^{1/(\tau+1)}W_{\bszero}^{-1},\quad 
  \bar{L}_{\bszero}^{1/(\tau^{-1}+1)} 
    = \bar{W}_{\bszero}\Lambda^{-1/(\tau^{-1}+1)}\bar{W}_{\bszero}^{-1}. 
\]
For convenience, we introduce the auxiliary operators 
\[
  M_{\bszero} = W_{\bszero}sW_{\bszero}^{-1},\quad 
  \bar{M}_{\bszero} = \bar{W}_{\bszero}s\bar{W}_{\bszero}^{-1}, 
\]
which are the initial values of the Orlov-Schulman operators \cite{TT95}
\[
\begin{gathered}
  M = W\left(s + \sum_{k=1}^\infty kt_k\Lambda^k\right)W^{-1},\\
  \bar{M} = \bar{W}\left(s - \sum_{k=1}^\infty k\bar{t}_k\Lambda^{-k}
            \right)\bar{W}^{-1}. 
\end{gathered}
\]

\begin{lemma}
\beq
  L^{1/(\tau+1)}_{\bszero} = q^{M_{\bszero}}q^{-s}\Lambda^{1/(\tau+1)},\quad 
  q^{M_{\bszero}} = q^s(1 - q^{(\tau+1)s - \tau - 3/2}\Lambda^{-1}). 
  \label{L0M0}
\eeq
\end{lemma}

\begin{proof}
Let us split the expression (\ref{W0}) of $W_{\bszero}$ 
into two parts: 
\[
  W_{\bszero} = Vq^{-(\tau+1)(s-1/2)^2/2},\quad 
  V = q^{(\tau+1)(s-1/2)^2/2}\prod_{i=1}^\infty(1 - q^{i-1/2}\Lambda^{-1}),  
\]
and compute $L_{\bszero}^{1/(\tau+1)} 
= W_{\bszero}\Lambda^{1/(\tau+1)}W_{\bszero}^{-1}$ step by step. 
The first step is to use the general formula 
\[
  \Lambda^{1/(\tau+1)}f(s) = f(s + 1/(\tau+1))\Lambda^{1/(\tau+1)}
\]
to rewrite the product of middle three terms as 
\[
\begin{aligned}
  & q^{-(\tau+1)(s-1/2)^2/2}\Lambda^{1/(\tau+1)}q^{(\tau+1)(s-1/2)^2/2}\\
  &= q^{-(\tau+1)(s-1/2)^2/2}q^{(\tau+1)(s+1/(\tau+1)-1/2)^2/}\Lambda^{1/(\tau+1)}\\
  &= q^{-\tau/(2(\tau+1))}q^s\Lambda^{1/(\tau+1)}. 
\end{aligned}
\]
In the next step, we handle the last two operators 
$q^s$ and $\Lambda^{1/(\tau+1)}$ separately: 
\[
  W_{\bszero}\Lambda^{1/(\tau+1)}W_{\bszero}^{-1} 
  = q^{-\tau/(2(\tau+1))}\cdot Vq^sV^{-1}\cdot V\Lambda^{1/(\tau+1)}V^{-1}. 
\]
Since 
\[
  Vq^sV^{-1} = W_{\bszero}q^sW_{\bszero} = q^{M_{\bszero}} 
\]
and 
\[
\begin{aligned}
  V\Lambda^{1/(\tau+1)}V^{-1}
  &= q^{(\tau+1)(s-1/2)^2/2}\Lambda^{1/(\tau+1)}q^{-(\tau+1)(s-1/2)^2/2}\\
  &= q^{(\tau+1)(s-1/2)^2/2}q^{-(\tau+1)(s+1/(\tau+1)-1/2)^2/2}\Lambda^{1/(\tau+1)}\\
  &= q^{\tau/(2(\tau+1))}q^{-s}\Lambda^{1/(\tau+1)}, 
\end{aligned}
\]
we obtain the expression of $L_{\bszero}^{1/(\tau+1)}$ in (\ref{L0M0}).  
Let us proceed to computation of $q^{M_{\bszero}}$. 
We can now use the general formula 
\[
  f(\Lambda)q^s = q^sf(q\Lambda)
\]
as 
\[
\begin{aligned}
  & \prod_{i=1}^\infty(1 - q^{i-1/2}\Lambda^{-1})\cdot q^s\cdot
    \prod_{i=1}^\infty(1 - q^{i-1/2}\Lambda^{-1})^{-1} \\
  &= q^s\prod_{i=1}^\infty(1 - q^{i-3/2}\Lambda^{-1})\cdot
    \prod_{i=1}^\infty(1 - q^{i-1/2}\Lambda^{-1})^{-1}\\
  &= q^s(1 - q^{-1/2}\Lambda^{-1}). 
\end{aligned}
\]
Consequently, 
\[
\begin{aligned}
  q^{M_{\bszero}} 
  &= q^{(\tau+1)(s-1/2)^2/2}q^s(1 - q^{-1/2}\Lambda^{-1})q^{-(\tau+1)(s-1/2)^2/2}\\
  &= q^s - q^{s-1/2}q^{(\tau+1)(s-1/2)^2/2}\Lambda^{-1}q^{-(\tau+1)(s-1/2)^2/2}\\
  &= q^s - q^{s-1/2}q^{(\tau+1)(s-1/2)^2/2}q^{-(\tau+1)(s-3/2)^2/2}\Lambda^{-1}\\
  &= q^s(1 - q^{(\tau+1)s-\tau-3/2}\Lambda^{-1}). 
\end{aligned}
\]
\end{proof}

\begin{lemma}
\beq
  \bar{L}_{\bszero}^{-1/(\tau^{-1}+1)} 
  = q^{\bar{M}_{\bszero}}q^{\tau s - \tau - 1/2}\Lambda^{-1/(\tau^{-1}+1)},\quad 
  q^{\bar{M}_{\bszero}} = q^s(1 - q^{-(\tau+1)s+1/2}\Lambda). 
  \label{Lbar0Mbar0}
\eeq
\end{lemma}

\begin{proof}
These expressions can be derived in much the same way 
as the proof the previous lemma.  We omit the detail. 
\end{proof}

\subsection{End of proof of Theorem \ref{key-thm}}

We can see from (\ref{L0M0}) and (\ref{Lbar0Mbar0}), 
by straightforward computation, that 
\beq
  L_{\bszero}^{1/(\tau+1)} 
  = (1 - q^{(\tau+1)s-\tau-1/2}\Lambda^{-1})\Lambda^{1/(\tau+1)} 
  = - \bar{L}_{\bszero}^{-1/(\tau^{-1}+1)}. 
  \label{L0Lbar0-rel}
\eeq
This means that the algebraic relation (\ref{LLbar-rel}) 
is satisfied at the initial time $\bst = \bar{\bst} = \bszero$. 

This is enough to conclude that (\ref{LLbar-rel}) itself 
is satisfied at all time.  
Note that the both sides of (\ref{LLbar-rel}) satisfy  
the same Lax equations of the form 
\beq
  \frac{\rd\calL}{\rd t_k} = [B_k,\calL],\quad 
  \frac{\rd\calL}{\rd\bar{t_k}} = [\bar{B}_k,\calL],\quad 
  k = 1,2,\ldots.
  \label{calL-Laxeq}
\eeq
Therefore, by the uniqueness of solution 
in the initial value problem of these equations, 
the both sides of (\ref{LLbar-rel}) 
with the same initial value should be equal 
to each other throughout the time evolutions. 
This completes the proof of Theorem \ref{key-thm}. 

\begin{remark}
The Lax and Orlov-Schulman operators turn out 
to satisfy the algebraic relation 
\beq
  \left(q^{-M}L^{1/(\tau+1)}\right)^{-\tau} 
  = q^{(\tau+1)/2}q^{-\bar{M}}\bar{L}^{-1/(\tau^{-1}+1)} 
  \label{LM-rel}
\eeq
that supplements (\ref{LLbar-rel}).  These two relations 
form a pair of conditions that single out a solution 
of the 2D Toda hierarchy \cite{TT95}. 
(\ref{LM-rel}) can be derived by the same logic 
as the foregoing derivation of (\ref{LLbar-rel}) as follows.  
Let us rewrite (\ref{L0M0}) and (\ref{Lbar0Mbar0}) as 
\[
  q^{-M_{\bszero}}L^{1/(\tau+1)}_{\bszero} = q^{-s}\Lambda^{1/(\tau+1)},\quad
  q^{-\bar{M}_{\bszero}}\bar{L}_{\bszero}^{-1/(\tau^{-1}+1)} 
  = q^{\tau s - \tau - 1/2}\Lambda^{-1/(\tau^{-1}+1)}.
\]
Since the right hand side of these equations satisfy 
the operator identity 
\[
  \left(q^{-s}\Lambda^{1/(\tau+1)}\right)^{-\tau} 
  = q^{(\tau+1)/2}q^{\tau s - \tau - 1/2}\Lambda^{-1/(\tau^{-1}+1)},  
\]
we obtain the algebraic relation 
\[
  \left(q^{-M_{\bszero}}L_{\bszero}^{1/(\tau+1)}\right)^{-\tau} 
  = q^{(\tau+1)/2}q^{-\bar{M}_{\bszero}}\bar{L}_{\bszero}^{-1/(\tau^{-1}+1)}.  
\]
This means that (\ref{LM-rel}) is satisfied 
at $\bst = \bar{\bst} = \bszero$.  Since the both sides 
of (\ref{LM-rel}) satisfy the same Lax equations, 
(\ref{LM-rel}) holds at all time.  
\end{remark}

\begin{remark}
(\ref{LM-rel}) can be reduced to the linear relation 
\beq
  \log L - (\tau + 1)(M - 1/2)\log q 
  = \log\bar{L} + (\tau^{-1} + 1)(\bar{M} - 1/2)\log q 
  \label{logLM-rel}
\eeq
among the four operators $\log L,\log\bar{L},M,\bar{M}$. 
Let us note that these operator satisfy 
the canonical commutation relations 
\[
  [\log L,M] = [\log\bar{L},\bar{M}] = 1. 
\]
We can thereby use the Baker-Campbell-Hausdorff formula 
to rewrite $q^{-M}L^{1/(\tau+1)}$ and 
$q^{-\bar{M}}\bar{L}^{-1/(\tau^{-1}+1)}$ as 
\[
\begin{aligned}
  q^{-M}L^{1/(\tau+1)}
  &= \exp\left(- M\log q + \frac{\log L}{\tau+1} 
       + \frac{\log q}{2(\tau+1)} \right),\\
  q^{-\bar{M}}\bar{L}^{-1/(\tau^{-1}+1)} 
  &= \exp\left(- \bar{M}\log q - \frac{\log\bar{L}}{\tau^{-1}+1} 
       - \frac{\log q}{2(\tau^{-1}+1)} \right). 
\end{aligned}
\]
Plugging these expressions into (\ref{LM-rel}) yields 
the exponentiated form 
\[
\begin{aligned}
  &\exp\left(\log L - (\tau + 1)(M - 1/2)\log q\right)\\
  &= \exp\left(\log\bar{L} + (\tau^{-1} + 1)(\bar{M} - 1/2)\log q\right)
\end{aligned}
\]
of (\ref{logLM-rel}). Thus (\ref{logLM-rel}) implies (\ref{LM-rel}). 
Actually, we can derive (\ref{logLM-rel}) directly 
by the same method as the derivation of (\ref{LM-rel}). 
The fact that (\ref{logLM-rel}) is satisfied 
at the initial time $\bst = \bar{\bst} = \bszero$ 
is a consequence of the relations 
\begin{gather}
  \log L_{\bszero} =  (\tau + 1)(M_{\bszero} - 1/2)\log q 
    + \log\Lambda - (\tau + 1)(s - 1/2)\log q,
  \label{logL0}\\
  \log\bar{L}_{\bszero} = - (\tau^{-1} + 1)(\bar{M}_{\bszero} - 1/2) 
    + \log\Lambda - (\tau + 1)(s - 1/2)\log q. 
  \label{logLbar0}
\end{gather}
These relations are obtained by computing 
$W_{\bszero}\Lambda^\epsilon W_{\bszero}^{-1}$ 
and $\bar{W}_{\bszero}\Lambda^\epsilon\bar{W}_{\bszero}$ 
explicitly with the aid of (\ref{W0}) and (\ref{Wbar0}) 
and taking the derivative at $\epsilon = 0$. 
\end{remark}

\section{Integrable structures in cubic Hodge integrals}

\subsection{Reduced system of Lax equations}

Let $\calL$ denote the operator (\ref{LLbar-u}) 
that emerges as a consequence of 
the key algebraic relation (\ref{LLbar-rel}). 
This operator satisfies the Lax equations (\ref{calL-Laxeq}). 
Actually, each equation can be rewritten 
in two different forms as 
\[
\begin{gathered}
  \frac{\rd\calL}{\rd t_k} = [(L^k)_{\ge 0},\calL] 
    = - [(L^k)_{<0},\calL],\\
  \frac{\rd\calL}{\rd\bar{t}_k} = [(\bar{L}^{-k})_{<0},\calL] 
    = - [(\bar{L}^{-k})_{\ge 0},\calL]. 
\end{gathered}
\]
This implies that only a $\Lambda^{-\tau/(\tau+1)}$-term 
survives on the right hand side.  Thus the Lax equations 
can be reduced to equations of the form 
\beq
  \frac{\rd u}{\rd t_k} = F_k,\quad 
  \frac{\rd u}{\rd\bar{t}_k} = \bar{F}_k, 
  \label{u-evoleq}
\eeq
where 
\[
  F_k = u\left((L^k)_0 - (L^k)_0|_{s\to s-\tau/(\tau+1)}\right)
  = (L^k)_{-1} - (L^k)_{-1}|_{s\to s+1/(\tau+1)}
\]
and
\[
  \bar{F}_k   = (\bar{L}^{-k})_{-1}|_{s\to s+1/(\tau+1)} - (\bar{L}^{-k})_{-1}
  = u\left((\bar{L}^{-k})_0|_{s\to s-\tau/(\tau+1)} - (\bar{L}^{-k})_0\right). 
\]
$(A)_n$ denotes the coefficient of $\Lambda^n$ 
in the difference operator $A$. 

If $B_k$'s and $\bar{B}_k$'s have {\it local} expressions 
with respect to $u$, namely, depend on a finite number 
of the shifted $u$'s $u(s),u(s\pm 1),u(s\pm 2),\ldots$, 
so do $F_k$'s and $\bar{F}_k$'s.  In such a case, 
(\ref{u-evoleq}) is a system of evolution equations for $u$ 
in a genuine sense. (\ref{calL-Laxeq}) becomes 
a Lax representation thereof.  
The problem of {\it nonlocality} arises, e.g., 
when we attempt to construct fractional powers 
of a difference operator directly, namely, without recourse 
to the use of a dressing operator  (cf. Carlet's construction 
of the bigraded Toda hierarchy \cite{Carlet06}).  

We shall not pursue the problem of locality, 
and develop our consideration rather formally. 
As we show below, various integrable hierarchies emerge 
when $\tau$ takes rational values of particular forms.

\subsection{When $\tau$ is a positive integer}

Let us consider the case where $\tau$ is equal 
to a positive integer $N$.  In this case, 
$\calL$ is a fractional difference operator of the form 
\beq
  \calL = \Lambda^{1/(N+1)} - u\Lambda^{-N/(N+1)}. 
  \label{calL(N)}
\eeq
This is exactly the Lax operator 
of the Bogoyavlensky-Itoh-Narita system 
\cite{Narita82,Bog87,Itoh87} realized 
on the fractional lattice $(N+1)^{-1}\ZZ \subset \ZZ$. 
The case of $N = 1$ amounts to the usual Volterra lattice. 
Thus, as conjectured and partially proved 
by Dubrovin et al. \cite{DY1606,DLYZ1612,LZZ17}, 
integrable hierarchies of the Volterra type underlie 
the cubic Hodge integrals when $\tau$ is a positive integer.  

The $\bst$-flows generated by 
\[
  B_k = (\calL^{(N+1)k})_{\ge 0},\quad k = 1,2,\ldots, 
\]
can be identified with the genuine time evolutions 
of the Bogoyavlensky-Itoh-Narita system. 
Moreover, the $(N+1)$-st power of $\calL$ 
is a difference operator of the form 
\beq
  \calL^{N+1} 
  = L = (-1)^{N+1}\bar{L}^{-N} 
  = \Lambda + p_1 + \cdots + p_{N+1}\Lambda^{-N}, 
\eeq
and can be identified with the Lax operator 
of the bigraded Toda hierarchy of the type $(1,N)$. 
Since 
\[
   (\calL^{(N+1)k})_{<0} = (-1)^{(N+1)k}\bar{B}_{Nk}, 
\]
the $\bst$-flows coincide with part of the $\bar{\bst}$-flows 
up to sign factors.  The other negative flows $\bar{t}_k$, 
$k \not\equiv 0 \mod N$, are nonlocal. 

The problem of nonlocality can be avoided 
if we leave the 2D Toda hierarchy and treat 
the Bogoyavlensky-Itoh-Narita system 
as a reduction of the lattice KP hierarchy 
(aka the discrete KP hierarchy \cite{Dickey-book}) 
\beq
  \frac{\rd L}{\rd t_k} = [B_k,L], \quad k = 1,2,\ldots. 
  \label{latticeKP}
\eeq
The lattice KP hierarchy is simply a subset 
of the 2D Toda hierarchy that consists of 
the same Lax operator $L$ and the Lax equations 
with respect to to $\bst$.  $B_k$'s are defined 
by $L$ as $B_k = (L^n)_{\ge 0}$, hence local. 
The tau function $\calT(s,\bst,\bszero)$ restricted 
to $\bar{\bst} = \bszero$, which is a generating function 
of the one-partition Hodge integrals, becomes a tau function 
of the lattice KP hierarchy. 

In a limit as $N \to \infty$, the Bogoyavlensky-Itoh-Narita 
hierarchy turns into a continuous version \cite{Bog88,Itoh88}. 
The reduced Lax operator (\ref{calL(N)}) is replaced therein 
by a difference-differential operator of the form 
\beq
  \calL = \log\Lambda - u\Lambda^{-1} = \rd_s - ue^{-\rd_s}.
  \label{calL(infty)}
\eeq
In the same limit, the cubic Hodge integrals become 
the {\it linear} Hodge integrals that are related 
to the Hurwitz numbers of $\CC\PP^1$ (the ELSV formula) 
\cite{ELSV00}.  We can thus reconfirm our recent result 
\cite{Takasaki18b} that the continuous Bogoyavlensky-Itoh hierarchy 
underlies the Hurwitz numbers.  

Let us mention that almost the same difference-differential 
operator as (\ref{calL(infty)}) is used in Buryak and Rossi's 
new Lax representation of the intermediate long wave hierarchy 
\cite{BR18}. This fact is extremely significant, 
because the intermediate long wave hierarchy 
was proposed by Buryak as an integrable structure 
of the linear Hodge integrals \cite{Buryak13,Buryak15}. 
Buryak's approach is based on the Dubrovin-Zhang theory 
of integrable Hamiltonian PDEs.  Buryak and Ross's 
Lax representation is derived along the same line.

\subsection{When $\tau$ is a positive rational number}

Let us turn to the more general case where $\tau$ 
is a positive rational number, i.e., $\tau = b/a$ 
where $a$ and $b$ are positive coprime integers. 
In this case, $\calL$ becomes the following generalization 
of (\ref{calL(N)}): 
\beq
  \calL = \Lambda^{a/(a+b)}  - u\Lambda^{-b/(a+b)}. 
\eeq
This is a fractional difference operator defined 
on the lattice $(a+b)^{-1}\ZZ \subset \RR$. 
Its $(a+b)$-th power is a difference operator of the form 
\beq
  \calL^{a+b} = L^a = (-1)^{a+b}\bar{L}^{-b} 
  = \Lambda^a + p_1\Lambda^{a-1} + \cdots + p_{a+b}\Lambda^{-b}, 
\eeq
and can be identified with the Lax operator 
of the bigraded Toda hierarchy of the type $(a,b)$. 

Thus, just like the relation between the Volterra lattice 
and the Toda lattice \cite{KvM75}, the generalized 
Bogoyavlensky-Itoh-Narita system sits over 
the bigraded Toda hierarchy.  The flows of these systems 
are generated by the common generators $B_k,\bar{B}_k$, 
$k = 1,2,\ldots$ (apart from the problem of locality). 

Let us mention that these reduced systems exhibit 
the duality under the exchange 
\[
  a \leftrightarrow b,\quad 
  \bst \leftrightarrow \bar{\bst},\quad
  L\leftrightarrow\bar{L}.
\]
This duality stems from the symmetry of 
the Hodge integrals $G_{g\mu\bar{\mu}}(\tau)$ 
under the exchange 
\[ 
  \tau \leftrightarrow \tau^{-1},\quad 
  \mu \leftrightarrow \bar{\mu}.
\]

\subsection{When $\tau$ is a negative rational number}

The situation changes qualitatively 
when $\tau$ is a negative rational number $-b/a$. 
(Just like the previous case, $a$ and $b$ are 
assumed to be positive coprime integers.) 
This case is divided to two cases, namely, 
$a > b$ and $a < b$.  Since these cases can be 
interchanged by the aforementioned duality, 
let us focus our consideration on the first case. 

If $\tau = - b/a$ and $a > b$, $\calL$ comprises 
only positive powers of $\Lambda$: 
\beq
  \calL = \Lambda^{a/(a-b)} - u\Lambda^{b/(a-b)}. 
\eeq
Its $(a-b)$-th power, too, contains only positive powers: 
\beq
  \calL^{a-b} = L^a = (-1)^{a-b}\bar{L}^b 
  = \Lambda^a + p_1\Lambda^{a-1} + \cdots + p_{a-b}\Lambda^b.
  \label{calL(-b/a)}
\eeq
Consequently, every $a$-th flow in the $\bst$-space 
are stationary: 
\beq
  \frac{\rd\calL}{\rd t_{ka}} 
  = [B_{ka},\calL] 
  = [\calL^{k(a-b)},\calL] = 0, \quad 
  k = 1,2,\ldots. 
\eeq
This is reminiscent of a periodic reduction 
of the 2D Toda hierarchy \cite{UT84}, but there is no periodicity 
in the present setting (as far as $u \not= 0$). 
The dressing operator $W$ is $p$-periodic, 
i.e., $[W,\Lambda^p] = 0$, if and only if $L^p = \Lambda^p$.  
The same equivalence holds for $\bar{W}$ and $\bar{L}$. 

In a sense, this case may be thought of 
as the bigraded Toda hierarchy of ``the type $(a,-b)$''. 
It is, however, also possible to forget $\bar{L}$ and $\bar{\bst}$ 
and to consider the reduced system within the lattice KP hierarchy 
(\ref{latticeKP}).  

(\ref{calL(-b/a)}) shows that the reduced system 
is a lattice version of the Gelfand-Dickey hierarchy 
\cite{Frenkel95}
\footnote{Frenkel's formulation \cite{Frenkel95} 
uses the $q$-shift operator $\Lambda = q^{x\rd_x}$ 
rather than the shift operator $\Lambda = e^{\rd_s}$, 
but this is not an essential difference.  
It should be stressed that the discrete KdV hierarchy 
in the sense of Dubrovin et al. \cite{DLYZ1409,DY1606,DLYZ1612} 
is distinct from the lattice KdV hierarchy in the present context.  
The discrete KdV hierarchy considered therein is an alias 
of the Volterra hierarchy.  Its Lax operator comprises 
both positive and negative powers of $\Lambda$.}. 
Speaking more precisely, this is slightly different 
from the usual lattice Gelfand-Dickey hierarchy 
in the sense that the $i$-th powers of $\Lambda$ 
for $i < b$ are missing in (\ref{calL(-b/a)}). 
The usual $a$-th Gelfand-Dickey reduction 
of the lattice KP hierarchy is characterized 
by the condition that the $a$-th power of $L$ 
contains no negative powers of $\Lambda$: 
\[
  L^a = \Lambda^a + p_1\Lambda^{a-1} + \cdots + p_a. 
\]
Truncating the $\Lambda^i$-terms for $i < b$ 
is consistent with the Lax equations 
of the lattice KP hierarchy, hence 
yields a further reduction of the system.  
Detailed properties of these reduced systems 
of the lattice KP hierarchy remain to be studied. 
We shall return to this issue elsewhere. 

Lastly, let us examine the discrete series 
\[
  \tau = - N/(N+1),\quad a = N+1, \quad b = N,\quad 
   N = 1,2,\ldots. 
\]
The one- and two-partition Hodge integrals of this type 
are hidden in our recent work on the three-partition 
Hodge integrals \cite{NT19}.  This explains why 
we encountered the Gelfand-Dickey hierarchy 
in the usual sense (namely, a reduction 
of the usual KP hierarchy) therein. 
It is well known \cite{Dickey-book} 
that the Gelfand-Dickey reduction of the lattice KP hierarchy 
in the $s$-space is accompanied by 
the usual Gelfand-Dickey hierarchy in the $t_1$-space.  

For these values of $\tau$, $\calL$ comprises 
two positive integral powers of $\Lambda$: 
\beq
  \calL = \Lambda^{N+1} - u\Lambda^N. 
\eeq
Since $\calL = L^{N+1} = - \bar{L}^N$, 
this case is the closest to a periodic reduction 
($L^{N+1} = \bar{L}^{N+1} = \Lambda^{N+1}$) 
of the 2D Toda hierarchy.  In many aspects, 
this case is situated at the opposite end of the case 
where $\tau$ is a positive integer. 
Unlike the Bogoyavlensky-Itoh-Narita system, 
the main degrees of freedom is contained in $L$ and $\bar{L}$, 
and $u$ is rather an auxiliary field.

\subsection*{Acknowledgements}

This work is a byproduct of collaboration with 
Toshio Nakatsu on the topological vertex 
and the cubic Hodge integrals.   
This work is partially supported by the JSPS Kakenhi Grant 
JP18K03350.

\end{document}